\documentclass[11pt]{article}       

\usepackage{hyperref}

\usepackage{graphicx}
\usepackage{fullpage}
\usepackage{amsthm}

\usepackage[x11names]{xcolor}

\usepackage{amsmath} 
\usepackage{amsfonts} 
\usepackage{amssymb} 
\usepackage{xspace}   
\usepackage{relsize} 
\usepackage{fancyvrb} 
\usepackage{enumerate} 
\usepackage{enumitem} 
\usepackage{centernot} 
\usepackage{adjustbox} 
\usepackage{multirow} 
\usepackage{stmaryrd} 
\usepackage{units} 
\usepackage[normalem]{ulem} 
\usepackage{ifthen} 
\usepackage{stmaryrd}
\usepackage{mathtools}

\usepackage[external]{forest}
\usetikzlibrary{backgrounds}

\usepackage{listings}
\newtheorem{theorem}{Theorem} 
\newtheorem{corollary}[theorem]{Corollary} 
\newtheorem{claim}[theorem]{Claim} 

\newcommand{\margincomment}[1]%
{{\marginpar{{\footnotesize\begin{minipage}{0.75in}%
          \begin{flushleft}%
            {#1}%
          \end{flushleft}%
        \end{minipage}%
      }}}\ignorespaces}

\newcommand{\ignore}[1]{}

\colorlet{darkgreen}{green!45!black}

\newcommand{\etal}{et al.\xspace}

\newcommand{\mycase}[1]{\mbox{{\underline{Case #1}}:\/}}

\newcommand{\braced}[1]{{ \left\{ #1 \right\} }}
\newcommand{\angled}[1]{{ \left\langle #1 \right\rangle }}
\newcommand{\parend}[1]{{ \left(#1 \right) }}

\newcommand{\sfmath}[1]{\operatorname{\text{\normalfont\sffamily #1}}}
\newcommand{\scmath}[1]{\operatorname{\text{\normalfont\scshape #1}}}

\newcommand{\Queries}{{\cal{U}}}
\newcommand{\Keys}{{\cal{K}}}  

\newcommand{\myproblem}[1]{{\ensuremath{\scmath{#1}}}\xspace}
\newcommand{\threeWCST}{\myproblem{3wcst}}
\newcommand{\twoWCST}{\myproblem{2wcst}}
\newcommand{\twoWCSTAllLoc}{{\ensuremath{\twoWCST_{\scmath{loc}}}}\xspace}
\newcommand{\twoWCSTAllNil}{{\ensuremath{\twoWCST_{\scmath{nil}}}}\xspace}

\newcommand{\baralpha}{\alpha}
\newcommand{\barbeta}{\beta}

\newcommand{\depth}{\ensuremath{\sfmath{depth}}}

\newcommand\eps\varepsilon

\newcommand{\subr}[1]{{\textsc{\small #1}}}

\newcommand{\cost}[1]{{{\sfmath{cost}}(#1)}}

\newcommand{\query}{{q}}  
\newcommand{\altquery}{{r}} 
\newcommand{\compnode}[2]{\ensuremath{\angled{\query\,{#1}\,{#2}}}}
\newcommand{\leafnodekey}[1]{\ensuremath{\braced{#1}}}
\newcommand{\leafnodenonkey}[1]{\ensuremath{\parend{#1}}} 
\newcommand{\queries}[1]{\Queries_{#1}}  
\newcommand{\notakey}{{\bot}}

\newcommand{\leaves}[1]{\ensuremath{\sfmath{leaves}(#1)}\xspace}

\newcommand{\weight}[1]{{w_{#1}}} 

\newcommand{\half}{{\textstyle\frac{1}{2}}}

\newcommand{\onefifth}{{\textstyle\frac{1}{5}}}
\newcommand{\threefifths}{{\textstyle\frac{3}{5}}} 
\newcommand{\ninefifths}{{\textstyle\frac{9}{5}}} 
\newcommand{\twelvefifths}{{\textstyle\frac{12}{5}}}

\newenvironment{mycases}[1]
{ \begin{list}{}{%
      \setlength \labelsep {0.5em}
      \setlength \labelwidth {0pt}
      \setlength \listparindent {\parindent}
      \setlength \leftmargin {#1}
      \setlength \rightmargin {0pt}
      \setlength \itemsep {0.7ex}
      \setlength \parskip {2pt}
      \setlength \parsep {1pt}   
      \setlength \labelwidth {-0.5em}
      \setlength \topsep {0pt}
      \setlength \partopsep {0pt}
    }}
{ \end{list}                  } 

\newenvironment{outercases}
{ \begin{mycases}{0em} }
{ \end{mycases}                  } 

\newenvironment{innercases}
{ \begin{mycases}{1em} }
{ \end{mycases}                  } 

\newcommand{\myparagraph}[1]{\smallskip\noindent\textbf{#1}} 
\newcommand{\emparagraph}[1]{\smallskip\noindent\textit{#1}} 

\newcommand{\key}[1]{{\mathsf{#1}}}
\newcommand{\keyA}{{\key A}}
\newcommand{\keyB}{{\key B}}
\newcommand{\keyC}{{\key C}}
\newcommand{\keyD}{{\key D}}
\newcommand{\keyE}{{\key E}}

\newcommand{\treesize}{\footnotesize}

\newcommand{\leafcolor}{red!3}
\newcommand{\interiorcolor}{green!4}
\newcommand{\subtreecolor}{black!4}
\newcommand{\edgecolor}{black!66}

\newcommand{\subtreeparams}[2]{
  \forestset{
    subtree1/.style={
      anchor=north, 
      child anchor=north, 
      outer sep=0pt,
      rounded corners=0.7em, 
      shape=semicircle,
      inner sep=1pt,
      text depth=4pt,
      xscale=#1,
      content format={
        \noexpand\scalebox{#2}[1]{\forestoption{content}}
      },
    },
  }
  \forestset{
    subtree/.style={
      subtree1,
      draw={\edgecolor, thin, dotted}, 
      edge={\edgecolor, thin}, 
      fill=\subtreecolor, 
    },
  }
  \forestset{
    dimtree/.style={
      subtree1,
      draw={gray, thin, dotted},
      edge={gray, thin, dotted}, 
      fill=gray!3,
      text=black!50, 
    },
  }
  \forestset{
    greentree/.style={
      subtree1,
      draw={\edgecolor, thin}, 
      edge={\edgecolor, thin}, 
      fill=\interiorcolor, 
    },
  }
}
\newcommand{\defaultsubtreeparams}{
  \subtreeparams{0.91}{1.1}
}
\defaultsubtreeparams 

\forestset{
  interior/.style = {
    draw={\edgecolor, thin}, 
    edge={\edgecolor, thin}, 
    rounded rectangle, 
    inner sep=1.8pt, 
    text height=1.6ex, 
    text depth=.5ex, 
    fill=\interiorcolor,
    anchor=north, 
  },
  leaf/.style = {
    draw={\edgecolor, thin}, 
    edge={\edgecolor, thin}, 
    rectangle, 
    inner sep=1.7pt, 
    text height=1.7ex, 
    text depth=.3ex, 
    child anchor=north, 
    fill=\leafcolor,
    anchor=north, 
  }, 
  nodraw/.style = {draw=none, fill=none, anchor=north}, 
  equalitytest/.style = {calign=last}, 
  every leaf node/.style={
    if n children=0{#1}{}
  },
  every interior node/.style={
    if n children=0{}{#1}
  },
}

\tikzstyle{edgeCommon} = [circle, font=\tiny, inner sep=0pt, outer sep=1.5pt] 
\tikzstyle{edgeYes} = [edgeCommon, text depth=0ex, text height=1.2ex, auto=right, node contents={y}, pos=0.35]
\tikzstyle{edgeNo} = [edgeCommon, text depth=0ex, text height=1.2ex, auto=left, node contents={n}, pos=0.4, outer sep=1pt] 
\tikzstyle{edgeZero} = [edgeCommon, auto=right, node contents={0}, pos=0.7] 
\tikzstyle{edgeOne} = [edgeCommon, auto=left, node contents={1}, pos=0.7]
\tikzstyle{edgeLess} = [edgeCommon, auto=right, node contents={$<$}]
\tikzstyle{edgeGreater} = [edgeCommon, auto=left, node contents={$>$}, pos=0.5]
\tikzstyle{edgeEqual} = [edgeCommon, auto=right, node contents={$=$}, pos=0.4]

\tikzset{label distance=-1pt}

\forestset{
  bright/.style={
    fill=yellow!12,
    draw={black, thick}, 
    edge=thick,
  },
  dim/.style={
    fill=gray!3,
    draw={gray, thin, dotted}, 
     text=black!50, 
  },
  dimedge/.style={
    edge={thin, dotted}
  },
  gbst/.style={
    l sep=0ex, 
    s sep=1em, 
    draw={\edgecolor, thin}, 
    edge={\edgecolor, thin}, 
    rectangle, 
    rounded corners=.6em, 
    fill=brown!15, 
    inner sep=.21em, 
    outer sep=0pt,
    content format={
      \noexpand\ensuremath{\noexpand\begin{array}{@{}c@{}}%
                 \forestoption{content}%
                 \noexpand\end{array}\noexpand}%
             },
           }
}

\colorlet{myblue}{black}          

\begin{document}

\title{On the Cost of Unsuccessful Searches in Search Trees with Two-way Comparisons}


\author{
  Marek Chrobak\thanks{University of California at Riverside. Research supported by NSF grants CCF-1217314 and CCF-1536026}
  \and
 Mordecai Golin\thanks{Hong Kong University of Science and Technology. Research funded by HKUST/RGC grant FSGRF14EG28 and RGC CERG Grant 16208415.}
 \and
J.~Ian Munro\thanks{University of Waterloo. Research funded by NSERC Discovery Grant 8237-2012 and the Canada Research Chairs Programme.}
\and
Neal E.~Young\thanks{University of California at Riverside. Research supported by NSF grant IIS-1619463.}
}

\maketitle






\textbf{}
\begin{abstract}
Search trees are commonly used to implement access operations to a set of stored keys. If this set is static and 
    the probabilities of membership queries are known in advance, then one can precompute an optimal search tree, namely one that minimizes 
    the expected access cost. For a non-key query, a search tree
    can determine its approximate location by returning the inter-key interval containing the query.
  This is in contrast to other dictionary data structures, 
  like hash tables, that only report a failed search.
  We address the question ``what is the additional cost of determining  
  approximate locations for non-key queries''?  We prove that for two-way comparison trees
  this additional cost is at most $1$.  Our proof is based on a novel
  probabilistic argument that involves converting a search tree
  that does not identify non-key queries into a random tree that does.
\end{abstract}


\section{Introduction}%
\label{sec: introduction}


Search trees are among the most fundamental data structures in computer science.
They are used to store a collection of values, called \emph{keys}, and
allow efficient access and updates. The most common operations
on such trees are search queries, where a search for a given query value $q$
needs to return the pointer to the node representing $q$, provided that $q$ is among 
the stored keys. 

In scenarios where the keys and the probabilities of all potential queries are fixed and known in advance,
one can use a \emph{static} search tree, optimized so that
its expected cost for processing search queries is minimized. 
These trees have been studied since the 1960s, including a classic work by
Knuth~\cite{Knuth1971} who developed an $O(n^2)$ dynamic programming
algorithm for trees with three-way comparisons (\threeWCST's).
A three-way comparison ``$q:k$'' between a query value $q$ and a key $k$ has three
possible outcomes: $q < k$, $q = k$, or $q > k$, and thus it may require
two comparisons,  namely ``$q=k$'' and ``$q < k$'',
when implemented in a high-level programming language. This was in fact pointed out by
Knuth himself in the second edition of \emph{``The Art of Computer Programming''}~\cite[\S6.2.2~ex.~33]{Knuth1998}.  
It would be more efficient to have each comparison in the tree correspond to just  one binary comparison.
Nevertheless, trees with two-way comparisons (\twoWCST's) are not as well
understood as \threeWCST's, and the fastest algorithm for computing such optimal trees runs 
in time $\Theta(n^4)$~\cite
{Anderson2002,chrobak_etal_optimal_search_trees_2015,chrobak2015optimal_erratum,chrobak_etal_simple_bcst_algorithm_2019}.

Queries for keys stored in the tree are referred to in the literature as \emph{successful} queries, 
while queries for non-key values are \emph{unsuccessful}.  Every \threeWCST inherently supports both types of queries.
The search for a non-key query $q$ in a \threeWCST determines the ``location'' of $q$
--- the inter-key interval containing $q$.
(By  an \emph{inter-key interval} we mean an inclusion-maximal open interval not containing any key.)
Equivalently, it returns $q$'s \emph{successor} in the key set (if any).
This feature is a by-product of 3-way comparisons --- even if this information is not needed,
the search for $q$ in a \threeWCST produces this information at no additional cost.
In contrast, other commonly used dictionary data structures
(such as hash tables) provide only one bit of information for non-key queries --- that the query is not a key.
This suffices for some applications, for example in parsing,
where one needs to efficiently identify keywords of a programming language. 
In other applications, however, returning the non-key query interval (equivalently, the successor) is useful.
For example, when search values are perturbed keys (say, obtained from inaccurate measurements),
identifying the keys nearest to the query may be important.

With this in mind, it is reasonable to consider two variants of \twoWCST's:
\twoWCSTAllLoc's, which are two-way comparison search trees that return the inter-key interval
of each non-key query (just like \threeWCST's), and \twoWCSTAllNil's, that only return
the ``not a key'' value $\notakey$ to report unsuccessful search (analogous to hash tables). 
Since \twoWCSTAllNil trees provide less information, they can cost less than \twoWCSTAllLoc's. To see why,
consider an  example (see Figure~\ref{fig: example 1}) with keys $\Keys = \{1,2\}$, each with probability
$\nicefrac{1}{5}$. Inter-key intervals $(-\infty,1)$, $(1,2)$, $(2,\infty)$ each have probability $\nicefrac{1}{5}$ as well.  
The optimum \twoWCSTAllLoc tree (which must determine the inter-key 
interval of each non-key query), has cost $\nicefrac{12}{5}$,
while the optimum \twoWCSTAllNil tree (which need only identify non-keys as such) has cost $\nicefrac{9}{5}$.
Note that \twoWCSTAllLoc trees 
are much more constrained; they contain exactly $2n+1$ leaves. \twoWCSTAllNil trees may contain between $n+1$ and $2n+1$ leaves.
(More precisely, these statements hold for \emph{non-redundant} trees --- see Section~\ref{sec: preliminaries}.)

To our knowledge, the first systematic study of \twoWCST's was conducted by
Spuler~\cite{Spuler1994Paper,Spuler1994Thesis}, whose definition matches our definition of \twoWCSTAllNil's.
Prior to that work, Andersson~\cite{andersson_not_on_searching_91} presented some experimental results 
in which using two-way comparisons improved performance. Earlier, various other types of 
search trees called \emph{split trees}, which are essentially restricted variants of \twoWCST's,
were studied in~\cite{Sheil1978,Huang1984,Perl1984,Hester1986,StephenHuang1984}.
(As pointed out in~\cite{chrobak_etal_huangs_algorithm_2018}, the results
in~\cite{Spuler1994Paper,Spuler1994Thesis,StephenHuang1984}  contain some fundamental errors.)


\begin{figure}[t] 
   \centering
  \includegraphics[width=4in]{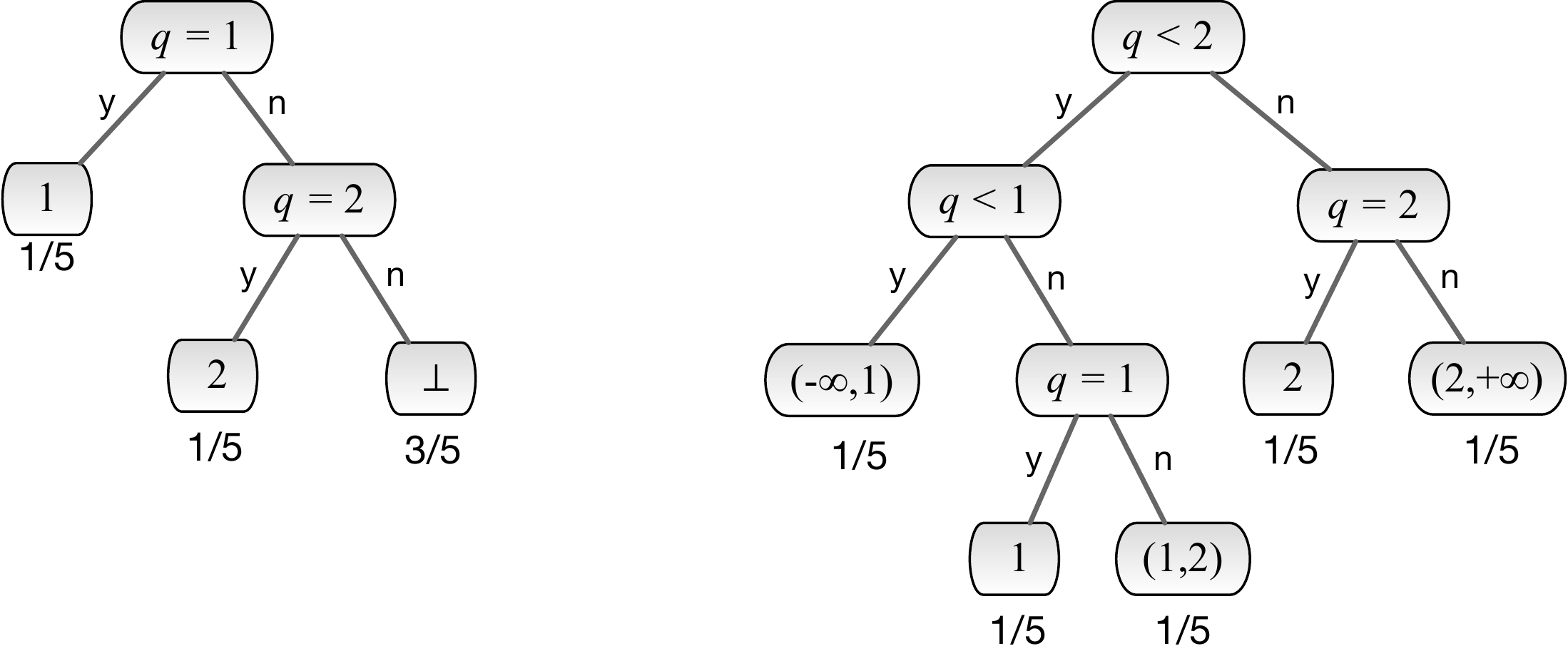}
  \caption{Identifying non-keys can cost more. 
  In this example, all keys and inter-key intervals have probability $\onefifth$.
  The cost of the \twoWCSTAllNil
	on the left is $1\cdot \onefifth + 2 \cdot \onefifth + 2\cdot \threefifths = \ninefifths$.
	The cost of the \twoWCSTAllLoc on the right is
	$2\cdot \onefifth + 3 \cdot \onefifth + 3 \cdot \onefifth + 2\cdot \onefifth +  2\cdot \onefifth = \twelvefifths$.
	We use the standard convention for graphical representation of search trees, with queries in the internal
	nodes, and with search proceeding to the left child if the answer to the query is ``yes'' and to the
	right child if the answer is ``no''.
}\label{fig: example 1}
\end{figure}



\emparagraph{Our contribution.}
The discussion above leads naturally to the following question:
``for two-way comparison search trees,
what is the additional cost of returning locations for all non-key queries''?
We prove that this additional cost is at most $1$.
Specifically (Theorem~\ref{theorem: gap bound} in Section~\ref{sec: gap bound}),
for any \twoWCSTAllNil $T^\ast$ there is
a \twoWCSTAllLoc $T'$ solving the same instance, such that the 
(expected) cost of a query in $T'$ is at most $1$ more than in $T^\ast$. 
We find this result to be somewhat counter-intuitive, since, as illustrated in Figure~\ref{fig: example 1},
a leaf in a \twoWCSTAllNil  may represent queries from multiple (perhaps even all) inter-query
intervals, so in the corresponding \twoWCSTAllLoc it needs to be split into multiple
leaves by adding inequality comparisons, which can significantly increase the tree depth.
The proof uses a probabilistic construction that converts $T^\ast$ into a random \twoWCSTAllLoc $T'$,
  increasing the depth of each leaf by at most one in expectation.

The bound in Theorem~\ref{theorem: gap bound} is tight.
To see why, consider an example with just one key $1$ whose probability is 
$\epsilon \in (0,\half)$ and inter-key intervals  $(-\infty,1)$ and $(1,\infty)$ having probabilities 
$\epsilon$ and $1-2 \epsilon$, respectively.
The optimum \twoWCSTAllNil has cost $1$, while the optimum \twoWCSTAllLoc has cost $2- \epsilon$. 
Taking $\epsilon$ arbitrarily close to $0$ establishes the gap.
(As in the rest of the paper, in this example the allowed comparisons are ``$=$'' and ``$<$'',
but see the discussion in Section~\ref{sec: Final Comments}.)


\emparagraph{Successful-only model.}
Many authors have considered \emph{successful-only} models, in which the trees support key queries but not non-key queries.
For \threeWCST's, Knuth's algorithm~\cite{Knuth1971} can be used for both the all-query and successful-only variants.
For \twoWCST's, in successful-only models the distinction between \twoWCSTAllLoc  and \twoWCSTAllNil does not arise. 
 Alphabetic trees can be considered as \twoWCST trees, in the successful-only model, 
restricted to using only ``$<$'' comparisons.  They can be built in $O(n \log n)$ time~\cite{Hu1971,Garsia1977}.
Anderson~{\etal}~\cite{Anderson2002}
gave an $O(n^4)$-time algorithm for successful-only \twoWCST's that use ``$<$'' and ``$=$'' comparisons.
With some effort, their algorithm can be extended to handle non-key queries too.
A simpler and equally fast algorithm, handling all 
variants of \twoWCST's (successful only, \twoWCSTAllNil, or \twoWCSTAllLoc's) was recently
described in~\cite{chrobak_etal_simple_bcst_algorithm_2019}.  


\emparagraph{Application to entropy bounds.}
For most types of two-way comparison search trees, the entropy of the distribution
is a lower bound on the optimum cost.
This bound has been widely used, for example to analyze approximation algorithms
(e.g.~\cite{mehlhorn_nearly_1975,Yeung1991,chrobak_etal_optimal_search_trees_2015,chrobak2015optimal_erratum}).
Its applicability to ``20-Questions''-style games (closely related to constructing \twoWCST's) was recently studied 
by Dagal~{\etal}~\cite{Dagan_etal_twenty_short_questions_stoc_2017,Dagan_etal_twenty_short_questions_journal_2019}.
But the entropy bound does not apply directly to \twoWCSTAllNil's. 
Section~\ref{sec: Entropy} explains why,
and how, with Theorem~\ref{theorem: gap bound}, it can be applied to such trees.


  \emparagraph{Other gap bounds.}
  To our knowledge, the gap between \twoWCSTAllLoc's and \twoWCSTAllNil's has not previously been considered.
  But gaps between other classes of search trees have been studied.
  Andersson~\cite{andersson_not_on_searching_91}  observed that for any depth-$d$ \threeWCST,
  there is an equivalent \twoWCSTAllLoc of depth at most $d+1$.
  Gilbert and Moore~\cite{Gilbert1959} showed that for any successful-only \twoWCST
  (using arbitrary binary comparisons), there is one using only ``$<$'' comparisons
  that costs at most 2 more.  This was improved slightly by Yeung~\cite{Yeung1991}.
  Anderson et al.~\cite[Theorem 11]{Anderson2002} showed that
  for any successful-only \twoWCST that uses ``$<$'' and ``$=$'' comparisons,
  there is one using only ``$<$'' comparisons that costs at most 1 more.
  Chrobak et al.~\cite[Theorem 2]{chrobak_etal_optimal_search_trees_2015,chrobak2015optimal_erratum}  leveraged Yeung's result
  to show that for any \twoWCSTAllLoc (of any kind, using arbitrary binary comparisons)
  there is one using only ``$<$'' and ``$=$'' comparisons  that costs at most 3 more.
  The trees guaranteed to exist by the gap bounds in~\cite  {Gilbert1959,Yeung1991,Anderson2002,chrobak_etal_optimal_search_trees_2015,chrobak2015optimal_erratum} 
  can be computed in $O(n\log n)$ time, whereas the fastest algorithms known for computing
  their optimal counterparts take time $\Theta(n^4)$.



\section{Preliminaries}%
\label{sec: preliminaries}


Without loss of generality, throughout the paper 
assume that the set of keys is $\Keys = \braced{1,2,\ldots, n}$ (with $n\ge 0$) and
that all queries are from the open interval $\Queries = (0,n+1)$. 

In a \twoWCSTAllLoc $T$ each internal node
represents a comparison between the query value, denoted by $\query$, and a key $k\in\Keys$.
There are two types of comparison nodes: equality comparison nodes $\compnode = k$, and
inequality comparison nodes $\compnode < k$.   
Each comparison node in $T$ has two children, one left and one right,
that correspond to the ``yes'' and ``no'' outcomes of the comparison, respectively. 
For each key $k$ there is a leaf $\leafnodekey k$ in $T$ and for
each $i\in \{0, 1, \ldots, n\}$ there is a leaf identified by open interval $\leafnodenonkey{i,i+1}$.
For any node $N$ of $T$, the subtree of $T$ rooted at $N$
(that is, induced by $N$ and its descendants) is denoted $T_N$.

Consider a query $q\in \Queries$.
A search for $q$ in a  \twoWCSTAllLoc $T$ starts at the root node of $T$ and follows a path
from the root towards a leaf. At each step, if the current node is a comparison
$\compnode = k$ or $\compnode < k$, if the outcome is ``yes'' then
the search proceeds to the left child, otherwise it proceeds to the right child.
A tree is correct if each query $q\in \Queries$ reaches a leaf $\ell$ such that $q\in \ell$.
Note that in a  \twoWCSTAllLoc there must be a comparison node  $\compnode = k$ for each key $k\in\Keys$.

The input is specified by a probability distribution $(\baralpha, \barbeta)$ on queries,
where, for each key $k\in\Keys$, the probability that $q = k$ is $\beta_k$ and
for each $i \in\{0,1,\ldots, n\}$ the probability that $q\in (i,i+1)$ is $\alpha_i$.
As the set of queries is fixed, the instance is uniquely determined by $(\baralpha, \barbeta)$.
The cost of a given query $\query$ is the number of comparisons in a search for $\query$ in $T$,
and the cost of tree $T$, denoted $\cost T$, is the expected cost of a random query $\query$.
(Naturally, $\cost T$ depends on $(\baralpha,\barbeta)$,
but the instance is always understood from context, so is omitted from the notation.)   
More specifically, for any query $\query\in \Queries$,  let $\depth_T(\query)$ 
denote the \emph{query depth} of $\query$ in $T$ ---
the number of comparisons made by a search in $T$ for $\query$.    
Then $\cost T$ is the expected value of $\depth_T(\query)$, where random query $\query$
is chosen according to the query distribution $(\baralpha,\barbeta)$.

The definition of \twoWCSTAllNil's  is similar to \twoWCSTAllLoc's.
The only difference is that non-key leaves do not represent the inter-key interval of the query:
in a \twoWCSTAllNil, each leaf either represents a key $k$ as before,
or is marked with the special symbol $\notakey$, representing any non-key query.
A \twoWCSTAllNil may have multiple leaves marked $\notakey$,
and searches for queries in different inter-key intervals   may terminate at the same leaf.

Also, the above definitions permit any key (or inter-key interval, for  \twoWCSTAllLoc's)
to have more than one associated leaf, in which case the tree is \emph{redundant}.
Formally, for any node $N$, denote by $\queries N$ the set of query values whose
search reaches $N$. (For the root, $\queries N = \Queries$.)
Call a node $N$ of $T$ \emph{redundant} if $\queries N = \emptyset$.  Define tree $T$
to be \emph{redundant} if it contains at least one redundant node. 
There is always an optimal tree that is non-redundant:
any redundant tree $T$ can be made non-redundant, without increasing cost, by
splicing out parents of redundant nodes. (If $N$ is redundant, 
replace its parent $M$ by the sibling of $N$, removing $M$ and $T_N$.)
But in the proof of Theorem~\ref{theorem: gap bound} it is technically useful to allow redundant trees.

For any non-redundant \twoWCSTAllLoc tree $T$,
the cost is conventionally expressed in terms of leaf \emph{weights}:
each key leaf $N = \leafnodekey k$ has weight $\weight N = \beta_k$,
while each non-key leaf $N = \leafnodenonkey{i,i+1}$ has weight $\weight N = \alpha_i$.
In this notation, letting $\leaves{T}$ denote the set of leaves of $T$ 
and $\depth_T(N)$ denote the depth of a node $N$ in $T$, 
\begin{equation}
  \cost T \,=\,  \sum_{L \in\, \leaves T } \weight{L}\cdot\depth_T(L)
  .\label{eqn: cost function}
\end{equation}
But the proof of Theorem~\ref{theorem: gap bound} uses only the earlier definition of cost,
which applies in all cases (redundant or non-redundant, \twoWCSTAllLoc or \twoWCSTAllNil).



\section{The Gap Bound}%
\label{sec: gap bound}


This section proves Theorem~\ref{theorem: gap bound}, that the additional
cost of returning locations of non-key queries is at most $1$. The proof uses a
probabilistic construction that converts any \twoWCSTAllNil
(a tree that does not identify locations of unsuccessful queries)
into a random \twoWCSTAllLoc (a tree that does). The conversion increases
the depth of each leaf by at most 1 in expectation, so increases the tree cost by at most $1$ in expectation.


\begin{theorem}\label{theorem: gap bound}~~~%
  Fix some instance $(\baralpha,\barbeta)$.
  For any  \twoWCSTAllNil tree $T^\ast$ for $(\baralpha,\barbeta)$ there is a
  \twoWCSTAllLoc tree $T'$ for $(\baralpha,\barbeta)$ such that  $\cost{T'}\le \cost{T^\ast}+1$.
\end{theorem}


\begin{proof} 
  Let $T^\ast$ be a given \twoWCSTAllNil.  Without loss of generality assume that  $T^\ast$ is non-redundant. 
  We describe a randomized algorithm that
  converts $T^\ast$ into a random \twoWCSTAllLoc tree $T'$ for $(\baralpha,\barbeta)$ with expected cost
  $E[\cost{T'}] \le \cost{T^\ast} + 1$.  Since the average cost of a random tree $T'$ satisfies this inequality, 
  some \twoWCSTAllLoc tree $T''$ must exist satisfying  $\cost{T''}\le \cost{T^\ast}+1$, proving the theorem.
  Our conversion algorithm starts with $T = T^\ast$ and then gradually
  modifies $T$, processing it bottom-up, and eventually produces $T'$.

  For any key $k\in\Keys$,
  let $\ell_k$ denote the unique $\notakey$-leaf at which a  search for $k$ \emph{starting in the no-child of $\compnode = k$}
  would end. Say that a leaf $\ell$ has a \emph{break due to $k$}
  if $k$ separates the query set $\queries \ell$ of $\ell$;
  that is, $\exists\query,\query'\in \queries \ell$ with $\query<k<\query'$.
  Note that $\ell_k$ is the only leaf in the tree that can have a break due to $k$.

In essence,  the algorithm converts $T^\ast$ into a (random) \twoWCSTAllLoc tree $T'$  by removing the breaks one by one.
  For each equality test $\compnode = k$ in $T^\ast$, 
  the algorithm adds one less-than comparison node $\compnode < k$ near $\compnode = k$
  to remove any potential break in $\ell_k$ due to $k$.
  (Here, by ``near'' we mean that this new node becomes either the parent or a child or  a grandchild of $\compnode = k$.)
  This can increase the depth of some leaves.
  The algorithm adds these new nodes in such a way that, in
  expectation, each leaf's depth will increase by at most $1$ during the whole process.
  In the end, if $\ell$ is a $\notakey$-leaf that does not have any breaks, then
  $\queries \ell$ represents an inter-key interval. (Here we also use the assumption that $T^\ast$ is non-redundant.)
  Thus, once we remove all breaks from $\notakey$-leaves, we obtain a \twoWCSTAllLoc tree $T'$. 

    To build some intuition before we dive into a formal argument, let's consider a node
    $N = \compnode = \keyB$, where $\keyB$ is a key, and suppose that  leaf
    $\ell_{\keyB}$ has a break due to $\keyB$. The left child of $N$ is leaf $\braced \keyB$, and
    let $t$ denote the right subtree of $N$. We can modify the tree by creating a new node
    $N' = \compnode < \keyB$, making it the right child of $N$, with left and right subtrees of $N'$
    being copies of $t$ (from which redundant nodes can be removed). 
    This will split $\ell_{\keyB}$ into two copies and remove the break due to $\ell_{\keyB}$, as desired.
    Unfortunately, this simple transformation also can increase the depth of some leaves, and thus also the cost of the tree.

    In order to avoid this increase of depth, our tree modifications also involve some local rebalancing
    that compensates for adding an additional node. The example above will be handled  using a case analysis.
    As one case, suppose that the root of $t$ is a comparison node $M = \compnode \diamond \keyA$,
    where $\diamond\in\braced{<,=}$ is any comparison operator and $\keyA$ is a key smaller than $\keyB$.
    Denote by $t_1$ and $t_2$ the left and right subtrees of $M$.
    Our local transformation for this case is shown in Figure~\ref{fig:conversion cases}(b).
    It also introduces $N' = \compnode < \keyB$, as before, but makes it the \emph{parent} of $N$.
    Its left subtree is  $M$, whose left and right subtrees are $t_1$ and a copy of $t_2$.
    Its right subtree is $N$, whose right subtree is a copy of $t_2$. As can be easily seen, this modification
    does not change the depth of any leaves except for $\ell_\keyB$. It is also correct, because in the original tree a search for any query $r\ge \keyB$
    that reaches $N$ cannot descend into $t_1$.

    The full algorithm described below breaks the problem into multiple cases. Roughly, in cases
    when $\ell_{\keyB}$ is deep enough in the subtree $T^\ast_N$ of $T^\ast$ rooted at $N$, we show that $T^\ast_N$ can be rebalanced after splitting
    $\ell_{\keyB}$. Only when $\ell_{\keyB}$ is close to $N$ we might need to increase the depth of $T^\ast_N$.

  %


\begin{figure}[!t] 
   \centering
  \includegraphics[height=6.45in]{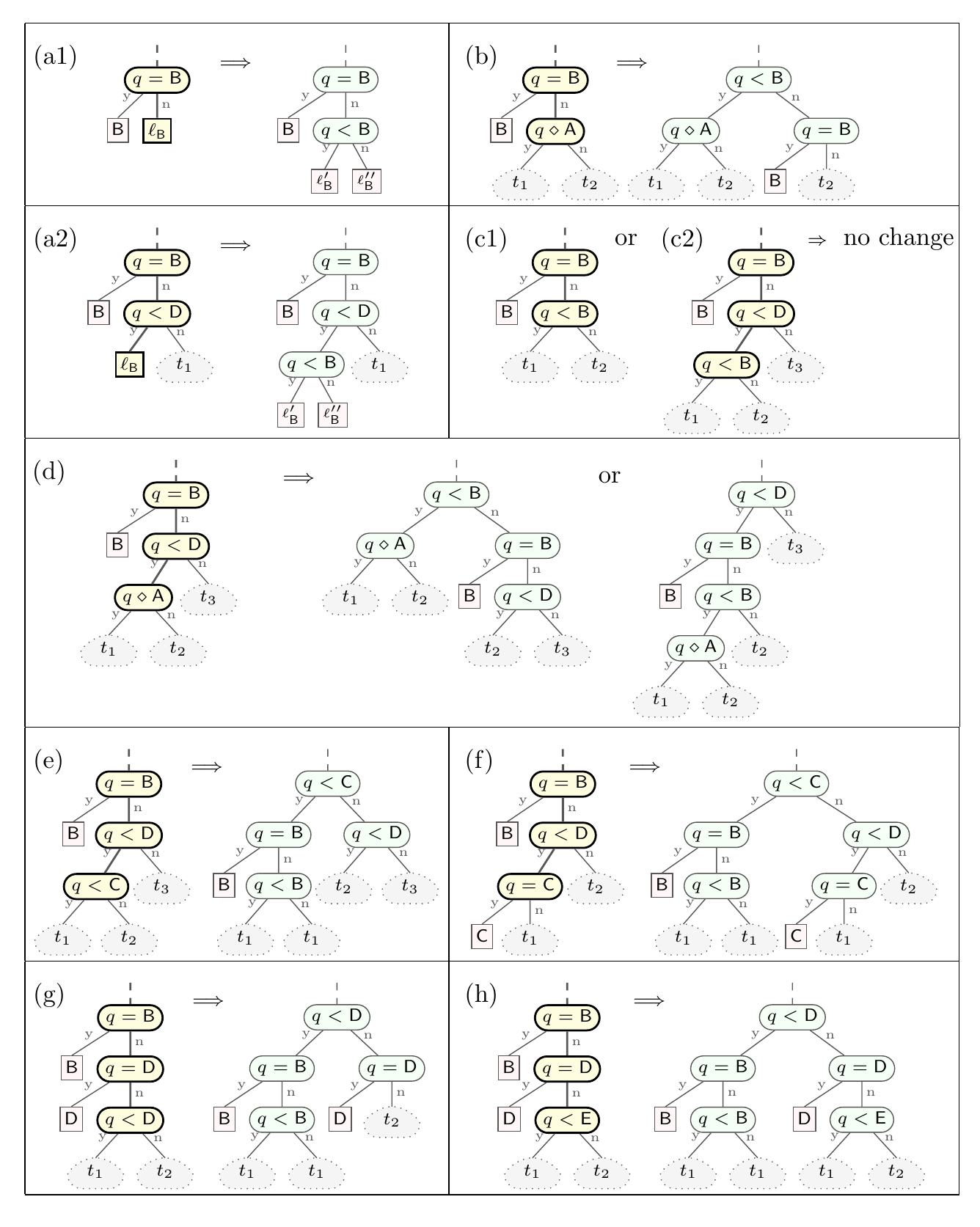}
    \caption{The ten possible types of an equality-node $N$, i.e.~$\compnode = \keyB$. For each
      type we show the conversion of its subtree $T_N$.
      Type (d) gives two possible replacements, and the algorithm chooses one randomly.
      The nodes on the prefix $P$ of the path from $N$ to leaf $\ell_\keyB$  have dark outlines.
      Along $P$, key $\keyA$ is the first key (if any) less than $\keyB$.
      Key $\keyD$ is the first key (if any) larger than $\keyB$.
      The second key (if any) larger than $\keyB$ is either $\keyC$ or $\keyE$.
      In types (b) and (d), symbol $\diamond$ is a comparison operator, $\diamond\in\braced{<,=}$.
      In cases (a1) and (a2), leaf $\ell_\keyB$ is split into two leaves
        $\ell'_\keyB$ and $\ell''_\keyB$, with appropriately modified query sets.
        In cases (b) and (d)-(h), the copies of $\ell_\keyB$ are in the duplicated sub-subtree $t_i$ of $\compnode = \keyB$.
    }\label{fig:conversion cases}
\end{figure}


  \smallskip


  \myparagraph{Conversion algorithm.}
  The algorithm processes all nodes in $T^\ast$ bottom-up, via a post-order node traversal,
  doing a conversion step \subr{Convert$(N)$} on each equality-test node $N$ of $T^\ast$.
  (Post-order traversal is necessary for the proof of correctness and analysis of cost.)
  More formally, the algorithm starts with $T = T^\ast$ and executes \subr{Process}$(T)$,
  where \subr{Process}$(T_N )$ is a recursive procedure
  that modifies the subtree rooted at node $N$ in the current tree $T$ as follows:
  \smallskip { 
    \\ \subr{Process}$(T_N)$:
    \\ \hspace*{1em} For each child $N_2$ of $N$, do \subr{Process$(T_{N_2})$}.
    \\ \hspace*{1em} If $N$ is an equality-test node, \subr{Convert}$(N)$.
  } 
  
  \smallskip
  \noindent
  (By definition, if $N$ is a leaf then  \subr{Process}$(T_N )$ does nothing.)

  \smallskip

  Procedure \subr{Process}$()$ will create copies of some subtrees and, as a result, it
  will also create redundant nodes in $T$. This might seem unnatural and wasteful, but it
  streamlines the description of the algorithm and the proof.
  Once we construct the final tree $T'$, these redundant nodes can be removed from $T'$
  following the method outlined in Section~\ref{sec: preliminaries}.

  \medskip\noindent
  Subroutine \subr{Convert}$(N)$, where $N$ is an equality-test node $\compnode = \keyB$,
  has three steps:
  {

    \begin{enumerate}
      \setlength{\itemsep}{5pt} 
    \item\label{step: first}\label{step: path} 
      Consider the path from $N$ to $\ell_\keyB$.
      Let $P$ be the prefix of this path that starts at $N$ and continues just until $P$ contains either
      {
        \begin{description}
          \setlength{\itemsep}{0pt}
      	\item{(i)} the leaf $\ell_\keyB$, or
        \item{(ii)} a second comparison to key $\keyB$, or
        \item{(iii)} any comparison to a key smaller than $\keyB$, or
        \item{(iv)} two comparisons to keys (possibly equal) larger than $\keyB$.  
        \end{description}
      }
      Thus, prefix $P$ contains $N$ and at most two other nodes.
      In case (iii), the last node on $P$ with comparison to a key smaller than $\keyB$ will be denoted
      $\compnode \diamond \keyA$, where  $\diamond \in \braced{=,<}$ is the comparison operation and	$\keyA$ is this key. 
      If $P$ has a comparison to a key larger than $\keyB$, denote the first such key by $\keyD$;
      if there is a second such key, denote it $\keyC$ if smaller than $\keyD$, or $\keyE$ if larger.
    \item\label{step: type}  
      Next, determine the \emph{type} of $N$. The type of $N$ is 
      whichever of the ten cases (a1)-(h) in Fig.~\ref{fig:conversion cases}  matches prefix $P$.
      (We show below that one of the ten must match $P$.) 
    \item\label{step: last}\label{step: convert} 
      Having identified the type of $N$, replace the subtree $T_N$ rooted at $N$ (in place)
      by the replacement for its type from Fig.~\ref{fig:conversion cases}. 

    \end{enumerate}
  }

  For example, $N$ is of type (b) if the second node $N_2$ on $P$ does a comparison to a key less than $\keyB$;
  therefore, as described in (\ref{step: path})\,(iii) above, $N_2$ is of the form $\compnode \diamond \keyA$.
  For type (b), the new subtree splits $P$ by adding a new comparison node $\compnode < \keyB$,
  with yes-child $N_2$ and no-child $N$, with subtrees copied appropriately from $T_N$.  
  (These trees are copied as they are, without removing redundancies. So after the reduction
  the tree will have two identical copies of $t_2$.)
  For type (d), there are two possible choices for the replacement subtree.
  In this case, the algorithm chooses one of the two uniformly at random.

  Intuitively, the effect of each conversion in Fig.~\ref{fig:conversion cases} is that leaf $\ell_\keyB$
    gets split into two leaves, one containing the queries smaller than $\keyB$ and the other containing the queries larger than $\keyB$.
    This is explicit in cases~(a1) and~(a2) where these two new leaves are denoted $\ell'_\keyB$ and $\ell''_\keyB$,
    and is implicit in the remaining cases.
  The two leaves resulting from the split may still contain other breaks, for keys of equality tests above $N$.
  (If it so happens that $\keyB$ equals $\min\,\queries{\ell_\keyB}$ or $\max\,\queries{\ell_\keyB}$,
  meaning that $\keyB$ is not actually a break,  then the query set of one of the resulting leaves will be empty.)

  This defines the algorithm.  Let $T' =$ \subr{Process}$(T^\ast)$ denote the random tree it outputs.
  As explained earlier, $T'$ may be redundant.


  \medskip
  \emparagraph{Correctness of the algorithm.}
  By inspection, \subr{Convert}$(N)$ maintains correctness
  of the tree while removing the break for $N$'s key $\keyB$, without introducing any new breaks.
  Hence, provided the algorithm completes successfully, the tree $T'$ that it outputs is a correct tree.
  To complete the proof of correctness,  we prove the following claim.


  \begin{claim}\label{cla: cases exhaustive}
    In each call to \subr{Convert}, some conversion (a1)-(h) applies.
  \end{claim}


  \begin{proof}
    Consider the time just before Step~\eqref{step: convert} of \subr{Convert}$(N)$.
    Let key $\keyB$, subtree $T_N$, and path $P$ be as defined for steps~\eqref{step: first}--\eqref{step: last} in 
    converting $N$.  Recall that $N$ is $\compnode = \keyB$. 
    Assume inductively that each equality-test descendant of $N$, when converted, had one of the ten types.
    Let $N_2$ be the second node on $P$, $N$'s no-child.  Let $N_3$ be the third node, if any.
    We consider a number of cases.
    
    \begin{outercases}
    \item[1. $N_2$ is a leaf] 
      Then $N$ is of type (a1).  

    \item[2. $N_2$ is a comparison node with key less than $\keyB$]
      Then $N$ is of type (b).
      
    \item[3. $N_2$ is a comparison node with key $\keyB$]
      Then $N_2$ cannot do an equality test to $\keyB$,
      because $N$ does that, the initial tree was irreducible,
      and no conversion introduces a new equality test. So $N$ is of type (c1).
      
    \item[4. $N_2$ is a comparison node with key larger than $\keyB$]
      Denote $N_2$'s key by $\keyD$. In this case $P$ has three nodes.  
      There are two sub-cases:

      \begin{innercases}

      \item[4.1. $N_2$ does a less-than test ($N_2$ is $\compnode < \keyD$)]
        By definition of $P$ and $\ell_\keyB$, the yes-child of $N_2$ is the third node $N_3$ on $P$.
        If $N_3$ is a leaf, then $N$ is of type (a2).  Otherwise $N_3$ is a comparison node.
        If $N_3$'s key is smaller than $\keyB$, then $N$ is of type (d).
        If $N_3$'s key is $\keyB$, then  $N$ is of type (c2).
        (This is because $\keyB$ has at most one equality node in $T_N$, as explained in Case~3.) 
        If $N_3$'s key is larger than $\keyB$ and less than $\keyD$, then $N$ is of type (e) or (f).

        To finish Case~4.1, we claim that \emph{$N_3$'s key cannot be $\keyD$ or larger}.
        Suppose otherwise for contradiction.  Let $N_3$ be $\compnode \diamond {\keyD'}$, where $\keyD'\ge \keyD$.
        By inspection of each conversion type,  no conversion produces an inequality root whose yes-child has larger key, 
        so $N_2$ was not produced by a previous conversion. So $N_2$ was in the original tree $T^\ast$, 
        where, furthermore, $N_2$'s yes-subtree contained a node with the key $\keyD'$.
        (This holds whether $N_3$ itself was in $T^\ast$,  or $N_3$ was produced by some conversion,
        as no conversion adds new comparison keys to its subtree.)
        This contradicts the irreducibility of $T^\ast$, proving the claim.

      \item[4.2. $N_2$ does an equality test ($N_2$ is $\compnode = \keyD$)]
        By the recursive nature of \subr{Process}$()$, the tree rooted at $N_2$ must be
        the result of applying \subr{Process}$()$ to the earlier no-child of $N = \compnode = \keyB$. 
        Further it  must be the  result of a  \subr{Convert}$()$ operation 
        (since \subr{Process}$()$ of an inequality comparison just returns that inequality comparison as root).
        Consider the previous conversion that produced $N_2$.  
        Inspecting the conversion types, the only conversions that could have produced $N_2$
        (with equality test at the root) are types (a1), (a2), (c1), and (c2).
        Each such conversion produces a subtree $T_{N_2}$
        where $N_2$'s no-child does some less-than test  $\compnode < X$ 
        to a key at least as large as the key of the root, that is $X \ge \keyD$. This node is now $N_3$.

        So, if $X = \keyD$, then $N$ is of type (g), while if $X>\keyD$, then $N$ is of type (h).
      \end{innercases}
    \end{outercases}\smallskip
    In summary, we have shown that at each step of our algorithm at least one of the
    cases in Fig.~\ref{fig:conversion cases} applies. 
    This completes the proof of Claim~\ref{cla: cases exhaustive}.
  \end{proof}


  \medskip
  \emparagraph{Cost estimate.}
  Continuing the proof of Theorem~\ref{theorem: gap bound}, we
  now estimate the cost of $T'$, the random tree produced by the algorithm. To prove $E[\cost{T'}] \le \cost{T^\ast} + 1$,
  we prove that, in expectation, the cost of each query $\altquery$ increases by at most 1.
  More precisely, we prove that for every query $\altquery\in \Queries$, 
  we have $E[\depth_{T'}(\altquery)] \le \depth_{T^\ast}(\altquery)+1$.

  Fix any query $\altquery\in \Queries$. 
  We distinguish two cases, depending on whether $\altquery$ is a key or not.

  \smallskip\noindent 
  \emph{\mycase{1. $\altquery\in \Keys$}}
  Then key $\altquery$ has one equality node $\compnode = \altquery$ in $T^\ast$.
  By inspection, each conversion (b) or (d)-(h) increases the query depth
  of the key $\keyB$ of converted node $\compnode = \keyB$ (i.e., $N$) by 1,
  and, in expectation, does not increase any other query depth.
  For example, consider a conversion of type (d).
  The depth of the root of subtree $t_1$ either increases by one or decreases by one,
  and, since each is equally likely, is unchanged in expectation.
  Likewise for $t_3$ and the first copy of $t_2$.
  The depth of the root of the second copy of $t_2$ is unchanged.
  Also, the queries $\altquery$ that descend into $t_2$ in $T_N$
  can be partitioned into those smaller than $\keyB$, and those larger.  For either random choice of replacement subtree,
  the former descend into the first copy of $t_2$,  the latter descend into the second copy.
  Hence, in expectation, if $\altquery = \keyB$ then this conversion increases the query depth of $\altquery$ by
  at most $1$,  and if  $\altquery\in \Queries -\braced{\keyB}$ then $\altquery$'s query depth does not increase.

  \smallskip

  By inspection of the two remaining conversion types, (a1) and (a2),
  each of those increases the depth of the queries in $\ell_\keyB$'s query set by 1,
  without increasing the query depth of any other query.
  Since $\altquery\in \Keys$, query $\altquery$ is not in leaf $\ell_{\keyB}$ for any such conversion.
  Hence, conversions (a1) and (a2) don't increase $\altquery$'s query depth.

  So  at most one conversion step in the entire sequence can increase $\altquery$'s query depth (in expectation)
  --- the conversion whose root is the equality-test node for $\altquery$,
  which increases the query depth by at most 1.
  It follows that the entire sequence increases the query depth of $\altquery$ by at most 1 in expectation.
  
  \smallskip\noindent 
  \emph{\mycase{2. $\altquery\not\in \Keys$}}
  In this case, $\altquery$ has no equality node in $T^\ast$.  As observed in Case~1, the only conversion step 
  that can increase the query depth of $\altquery$ (in expectation) is an (a1) or (a2) conversion of a node $\compnode = B$
  where $\ell_\keyB$ is $\altquery$'s leaf (that is, $r\in\Queries_\ell$). This step increases $\altquery$'s query depth by 1.
  
  So consider the tree just before such a conversion step applied to the subtree $T_N$, 
  where case~(a1) or~(a2) is applied and $\altquery$'s leaf is $\ell_\keyB$.
  We show the following property holds at that time:

  
  \begin{claim}\label{claim: no earlier step with leaf r}
    There was no earlier step whose conversion subtree contained the leaf of $\altquery$.
  \end{claim} 

  \begin{proof}
    To justify this claim, we consider cases~(a1) and (a2) separately.
    For case~(a1), $\altquery$'s leaf has \emph{no processed ancestors}. 
    (A ``processed'' node is any node in the replacement subtree of any previously implemented conversion.) 
    But there is no conversion type that produces such a leaf, proving the claim in this case.
    The argument in case~(a2) is a bit less obvious but similar: 
    in this case $\altquery$'s leaf is a yes-child and its parent is an inequality node that
    is the only processed ancestor of this leaf.
    By inspection of each conversion type, for each conversion that produces a leaf with
    only one processed ancestor (which would necessarily be the root for the
    converted subtree), this ancestor is either an equality test 
    (cases~(a1), (a2), (c1), (c2)), or has this leaf be a no-child of its
    parent (the second option of case~(d), with $t_3$ being a leaf).
    Thus no such conversion can produce a subtree of type~(a2) with $\altquery$'s 
    leaf being $\ell_\keyB$, completing the proof of the claim.
  \end{proof} 
  
  We then conclude that in this case ($\altquery\not\in \Keys$), there is at most one step
  in which the expected query depth of $\altquery$ can increase; and if it does, it increases
  only by $1$, so the total increase of  $\altquery$'s query depth is at most $1$ in expectation.

  \smallskip

  Summarizing, in either Case~1 or~2, the entire sequence of operations increases $\altquery$'s query depth
  by at most one in expectation (with respect to the random choices of the algorithm),
  that is $E[ \depth_{T'}(\altquery) ] \le \depth_{T^\ast}(\altquery) + 1$.
  Since this property holds for any $\altquery\in \Queries$, applying linearity of expectation
  (and using $\depth_T((i, i+1))$ to represent the depth in $T$ of queries in inter-key interval $(i, i+1)$),
  \begin{align*} 
    E[\cost{T'}]
    &{} \;=\; \textstyle
      E\Big[
      \sum_{i=1}^n \beta_i \depth_{T'}(i) 
      {} +
      \sum_{i=0}^n \alpha_i \depth_{T'}((i,i+1))
      \Big]
      \smallskip \\
    &{}\; = \;\textstyle
      \sum_{i=1}^n \beta_i\, E[\depth_{T'}(i)]
      {} +
      \sum_{i=0}^n \alpha_i\, E[\depth_{T'}((i,i+1))]
      \smallskip \\
    & {} \;\le\; \textstyle
      \sum_{i=1}^n \beta_i (1+ \depth_{T^{\ast}}(i))
      {} +
      \sum_{i=0}^n \alpha_i (1+ \depth_{T^{\ast}}((i,i+1)))
      \smallskip \\
    & {} \;=\;
      1 + \cost {T^\ast}.
  \end{align*}
  This completes the proof of Theorem~\ref{theorem: gap bound}.
\end{proof}




\section{Application To Entropy Bounds}%
\label{sec: Entropy}


In general, a search tree determines the answer to a query from a set of some number $m$ of possible answers.
In the successful-only model there are $n$ possible answers, namely the key values.
In the general $\twoWCSTAllLoc$ model 
there are $2n+1$ answers: the $n$ key values and the $n+1$ inter-key intervals. 
  In the $\twoWCSTAllNil$ model  there are $n+1$ answers: the $n$ key values and $\notakey$.
Let  $p$ be a probability distribution on the $m$ answers, namely $p_j$ is the probability that the answer
to a random query should be the $j$th answer.
It is well-known that any  binary-comparison search
tree $T$ that returns such answers in its leaves satisfies
$\cost{T} \ge  H(p)$, where $H(p) = \sum_j p_j \log_2 \frac 1 {p_j}$ is the \emph{Shannon entropy} of $p$.  
This fact is a main tool used for lower bounding the  optimal cost of search trees~\cite{ahlswede1987search}.

The entropy bound can be weak when applied directly to  \twoWCSTAllNil's.
To see why, consider a probability distribution $(\baralpha,\barbeta)$ on keys and inter-key intervals.
Since \twoWCSTAllNil's do not actually identify inter-key intervals, the answers associated with a \twoWCSTAllNil
are the key values and the $\notakey$ symbol representing the ``not a key'' answer, so the corresponding distribution is $(A,\barbeta)$,
for  $A = \sum_i \alpha_i$. Thus the entropy lower bound is
\begin{equation*}
  \cost{T^\ast} \;\ge\; H(A,\barbeta) = A \log_2 \frac 1 A + \sum_i \beta_i \log_2 \frac  1 {\beta_i}
\end{equation*}
for any \twoWCSTAllNil tree $T^\ast$. On the other hand,   by Theorem~\ref{theorem: gap bound}, $\cost{T^\ast} \ge \cost{T'} -1$
  for some \twoWCSTAllLoc tree $T'$.
  The entropy lower bound $\cost {T'} \ge H(\baralpha, \barbeta)$ applies to $T'$, giving the following lower bound:
  \begin{corollary}\label{cor}
    For any \twoWCSTAllNil tree $T^\ast$ for any input $(\baralpha, \barbeta)$,
    \(
    \cost{T^\ast}     \,\ge\,  H(\baralpha,\barbeta) -1.
    \)
  \end{corollary}

  To see that this bound can be stronger,  consider the following extreme example.
  Suppose that $\beta_k = 1/n^2$ for all $k$, and
  that $\alpha_i  = \frac{1}{n+1}  \left(1 - \frac{1}{n}\right)$ for all $i$.
  Then $A = 1 - \frac  1 n$,   $\sum_k \beta_k \log_2 \frac {1}{\beta_k} = \Theta(\log_2(n)/n)$,  
  and   $\sum_i \alpha_i \log_2 \frac {1}{\alpha_i} = \log_2 n - O(\log(n)/n)$.
  The direct entropy lower bound, $H(A, \barbeta)$, is
  \begin{equation*}
    A \log_2 \frac 1 A +  \sum_k \beta_k \log_2 \frac 1 {\beta_k} 
    \;=\; \Theta\Big(\frac{1}{n}\Big) + \Theta\Big(\frac{\log n}{n}\Big) 
    \;=\; o(1).
  \end{equation*}
  In contrast the lower bound in Corollary~\ref{cor} is
  \begin{equation*}
    -1 + \sum_i \alpha_i \log_2 \frac 1 {\alpha_i}   +   \sum_k \beta_k \log_2 \frac 1 {\beta_k}
    \;=\;
    \log_2(n) - o(1) - 1,
  \end{equation*}
  which is tight up to lower-order terms.

    Generally, the difference between the lower bound from Corollary~\ref{cor}
    and the direct entropy lower bound is $A\, H(\alpha/A) - 1$.    This is always at least $-1$.
    A sufficient condition for the difference to be large is that $A = \omega(1/\log n)$,
    with $\Omega(n)$ $\alpha_i$'s distributed more or less uniformly
    (i.e., $\alpha_i / A = \Omega(1/n)$),    so $H(\alpha/A) = \Theta(\log n)$.



\section{Final Comments}%
\label{sec: Final Comments}


  The proof of Theorem~\ref{theorem: gap bound} is quite intricate. It would be worthwhile to
  find a more elementary argument. We leave this as an open problem.

  We should point out that bounding the gap by a constant \emph{larger} than $1$ is considerably easier. 
  For example, one can establish a constant gap result by following the basic idea of our 
  conversion argument in Section~\ref{sec: gap bound} but using only a few simple rotations to
  achieve rebalancing. (The value of the constant may depend on the rebalancing strategy.)
  Another idea involves ``merging'' each key $k$ in $T^\ast$ and the adjacent failure interval $(k, k+1)$ into one
  key with probability $\beta_k+\alpha_k$, computing an optimal (successful-only) tree $T'$ for these new merged keys, and
  then splitting the leaf corresponding to this new key into two leaves, using an equality comparison.
  A careful analysis using the Kraft-Mcmillan inequality and the construction of alphabetic trees in~\cite[Theorem 3.4]{ahlswede1987search}
  shows that $\cost{T'} \le \cost{T^\ast} + 1$,
  proving a gap bound of 2.  (One reviewer of the paper also suggested this approach.)
  Reducing the gap to $1$ using this strategy does not seem possible though, as the second step
  inevitably adds $1$ to the gap all by itself.

  Theorem~\ref{theorem: gap bound} assumes that the allowed comparisons are ``$=$'' and ``$<$'', 
  but the proof can be extended to also allow comparison ``$\le$'' (that is, each comparison may be any of  $\{=,<,\le\}$)
  by considering a few additional cases in Figure~\ref{fig:conversion cases}. In the model with three comparisons, we do not know
  whether the bound of $1$ in  Theorem~\ref{theorem: gap bound} is tight.

  One other intriguing and related open problem is the complexity of computing optimum \twoWCST's. 
  The fastest algorithms in the literature for computing such optimal trees run 
  in time $\Theta(n^4)$~\cite{Anderson2002,chrobak_etal_optimal_search_trees_2015,chrobak2015optimal_erratum,chrobak_etal_simple_bcst_algorithm_2019}.
  Speed-up techniques for dynamic programming based on Monge properties or quadrangle inequality, now standard,
  were used to develop an $O(n^2)$ algorithm for computing optimal \threeWCST's~\cite{Knuth1971}.
  These techniques do not seem to apply to \twoWCST's, and new techniques would be needed to reduce the
  running time to $o(n^4)$.



  \bigskip
  \paragraph*{Acknowledgments}
  We are grateful to the anonymous reviewers for their numerous and insightful comments
  that helped us improve the presentation of our results.



\bibliographystyle{plainurl}
\bibliography{optimal_search_trees}


\end{document}